\documentclass[a4paper,11pt,oneside]{article}
\usepackage{latexsym}
\usepackage[english]{babel}
\usepackage{fancyhdr}
\usepackage[mathscr]{eucal}
\usepackage{amsmath}
\usepackage{mathrsfs}
\usepackage{amsthm}
\usepackage{amsfonts} 
\usepackage{amssymb} 
\usepackage{amscd}
\usepackage{bbm}
\usepackage{graphicx}
\usepackage{graphics}

\usepackage{bm}

\newcommand{\D}{\mathrm{d}}
\newcommand{\I}{\mathrm{i}}

\renewcommand{\vec}[1]{\boldsymbol{#1}}
 \newtheorem{theorem}{Theorem}[section]
 \newtheorem{corollary}[theorem]{Corollary}

 \theoremstyle{definition}
 
 \theoremstyle{remark}
 \newtheorem{remark}[theorem]{Remark}
 
 \numberwithin{equation}{section}

\begin{document}\hyphenation{Cou-lomb}

%
%
%
%
%
%
%
%
%


\title{Self-Adjoint Extensions of Dirac Operator with Coulomb Potential\footnote{This work was partially supported by the 2014-2017 MIUR-FIR grant ``\emph{Cond-Math: Condensed Matter and Mathematical Physics}'' code RBFR13WAET.}}
%
%

\author{Matteo Gallone\footnote{International School for Advanced Studies -- SISSA, via Bonomea 265, 34136 Trieste, Italy. e-mail 
\texttt{mgallone@sissa.it}}}

\maketitle

\abstract{In this note we give a concise review of the present state-of-art for the problem of self-adjoint realisations for the Dirac operator with a Coulomb-like singular scalar potential $V(\vec x)=\phi(\vec x) I_4$. We try to follow the historical and conceptual path that leads to the present understanding of the problem and to highlight the techniques employed and the main ideas. In the final part we outline a few major open questions that concern the topical problem of the multiplicity of self-adjoint realisations of the model, and which are worth addressing in the future.}


\section{Introduction}
In relativistic quantum mechanics one is interested in the study of Dirac equation, a partial differential equation that describes the dynamics of a $\frac{1}{2}$--spin fermion. The phase space of the physical system is the Hilbert space $\mathcal{L}^2:=L^2(\mathbb{R}^3,\mathbb{C}^4,\D^3x)$ which is
\begin{equation}
\mathcal{L}^2 := \left\{ u\, |\, u : \mathbb{R}^3 \to \mathbb{C}^4, \, \Vert u \Vert_{\mathcal{L}^2} < \infty \right\},
\end{equation}
where if $u=(u_1,u_2,u_3,u_4)$, with $u_j:\mathbb{R}^3 \to \mathbb{C}$, the norm is defined as $\Vert u \Vert_{\mathcal{L}^2}^2 = \int_{\mathbb{R}^3}\sum_{j=1}^4 |u_j(\vec x)|^2 \, \D^3 x$. 

The minimal Dirac operator is defined by
\begin{equation}
T = \boldsymbol\alpha \cdot \vec{p}+\beta +V(\vec x)
\end{equation}
on the compactly supported smooth functions:
\begin{equation}
\mathcal{D}(T)=\mathcal{C}^\infty_c :=C^\infty_c(\mathbb{R}^3\setminus\{0\};\mathbb{C}^4),
\end{equation}
where $\vec p=-\I \vec \nabla$ is the momentum operator, $ \boldsymbol\alpha = (\alpha_1,\alpha_2,\alpha_3)$, $\alpha_j$ and $\alpha_4=\beta$ are $4\times4$ Hermitian matrices which satisfy the anti-commutation relation
\begin{equation}
\alpha_j\alpha_k+\alpha_k\alpha_j=2 \delta_{jk} I_{4},
\end{equation}
$I_n$ is the $n\times n$ identity matrix and $V(x)$ is a real $4 \times 4$ matrix valued function called \emph{potential}. A standard form for the $\alpha$ matrices is the following
\begin{equation}
\alpha_j=\left(\begin{array}{cc} 0 & \sigma_j \\ \sigma_j & 0 \end{array}\right), \qquad \beta=\left(\begin{array}{cc}
I_2 & 0 \\ 0 & -I_2
\end{array}\right),
\end{equation}
where $\sigma_j$ are the Pauli matrices
\begin{equation}
\sigma_1=\left(\begin{array}{cc} 0 & 1 \\ 1 & 0 \end{array}\right),
\qquad
\sigma_2=\left(\begin{array}{cc} 0 & -i \\ i & 0 \end{array}\right),
\qquad
\sigma_3=\left(\begin{array}{cc} 1 & 0 \\ 0 & -1 \end{array}\right).
\end{equation}

In these notes we are interested in real scalar potentials of the form $V(\vec x)=\phi(\vec x) I_4$ which have a Coulomb-like singularity at the origin, namely $\lim_{\vec x \to 0} |\vec x| \phi(\vec x)= \nu \in \mathbb{R}$. For the sake of simplicity we will denote by $T_0$ the free Dirac operator and will refer to the operator $T=T_0+V$ as the \emph{Dirac-Coulomb} operator. A natural choice for $\phi$ is the Coulomb potential
\begin{equation}\label{eq:CoulombPot}
\phi(\vec x)= \frac{\nu}{|\vec x|},
\end{equation} 
and this means that one is modelling an electron subject to the electric field generated by $\nu$ positive charge in the origin. 

In atomic models $\nu$ is related to the atomic number by
\begin{equation}
\nu=\frac{Z}{\alpha},
\end{equation}
where $Z$ is the atomic number and $\alpha$ is the fine-structure constant $\alpha \sim 137$.

In the case of a multi-electron atom one can use some kind of screening approximation and an effective potential which is still Coulomb-like but it loses some properties like the spherical symmetry. This makes the study of self-adjoint extensions physically interesting also in the case of potentials with non spherical symmetry.

We collect in the following theorem what is known about the existence and uniqueness of self-adjoint extensions of the minimal Dirac-Coulomb operator \cite{Weid,S72,S72bis,KW79,X99,H13}.

\begin{theorem}[Self-adjoint extensions of the minimal Dirac-Coulomb operator]\label{th:1}
Let $T=T_0+V$ be a Dirac-Coulomb operator defined on $\mathcal{C}^\infty_c$ with $V(\vec x)=\phi(\vec x)I_4$ and
\begin{equation}
\lim_{\vec x \to 0} |\vec x| \phi(\vec x) = \nu.
\end{equation}
Then:
\begin{itemize}
\item[i) ]  if $|\nu|<\frac{\sqrt{3}}{2}$, then the operator $T$ is essentially self-adjoint and its unique self-adjoint extension has domain
\begin{equation}
\mathcal{D}(\bar T) = \mathcal{H}^1:=H^1(\mathbb{R}^3,\mathbb{C}^4,\D^3 x).
\end{equation}
\item[ii) ]  If $\frac{\sqrt{3}}{2} <|\nu| < 1$, then the operator $T$ has infinitely many self-adjoint extensions and if $\phi(\vec x)$ is bounded below or above there exists a unique distinguished extension $T_d$ with the properties
\begin{equation}\label{eq:IntrEn}
\mathcal{D}(T_d)\subset \mathcal{D}(|\vec x|^{-1/2}), \qquad \mathcal{D}(T_d)\subset \mathcal{D}(|T_0|^{1/2}).
\end{equation}
\item[iii) ]  If $|\nu|>1$, then there are infinitely many self-adjoint extensions of $T$.
\end{itemize}
\end{theorem}

The reason why in the regime $\frac{\sqrt{3}}{2} < |\nu| < 1$ we call $T_d$ distinguished is that physically the condition (\ref{eq:IntrEn}) is a requirement for the functions in $\mathcal{D}(T_d)$ to have a finite kinetic and potential energy. It is also notice-worthy that $T_d$ is the unique self-adjoint extension with this property (see section \ref{subsec:3.2}).

Another important remark is on the treshold value $|\nu|=\frac{\sqrt{3}}{2}$. In fact, in this case, it is not possible to determine whereas $T$ is essentially self-adjoint or not without any further information on $V$ (see \cite{X99} for more details). In the special case of the pure Coulomb potential (\ref{eq:CoulombPot}) for the choice $|\nu|=\frac{\sqrt{3}}{2}$ the operator $T$ is essentially self-adjoint.

Due to these reasons, in the literature one usually refers to i) as the \emph{regular regime}, to ii) as the \emph{critical regime}, and to iii) as the \emph{supercritical regime}.

The first step in the study of self-adjoint extensions is the computation of the deficency indices of the Dirac-Coulomb operator: we discuss it in section \ref{sec:DefIndex}. In section \ref{sec:History} we place the study of the self-adjoint extensions of Dirac-Coulomb operators into a historical perspective from Kato's paper in 1951 up to recent works. This includes also the sketch of some key proofs, with no pretension of completeness. In section \ref{sec:Hogeve} we present what is known about the classification of self-adjoint extensions of the Dirac operator. Last, in section \ref{sec:Future} we present some questions that, to our opinion, are more relevant and deserve further investigations.

\section{Deficency indices}\label{sec:DefIndex}
In this section we compute the deficency indices for the Dirac-Coulomb operator. We recall that given a densely defined symmetric operator $T$ its deficency indices are
\begin{equation}
n_{\pm}:=\mathrm{dim} \ker(T^* \mp i).
\end{equation}
In a sense they measure 'how far' the operator $T$ is from being self-adjoint. More precisely, by a well-known result (see \cite{RS} Corollary to Theorem X.2), a densely defined symmetric operator admits non-trivial self-adjoint extensions if and only if the deficency indices are equal and different from zero: $n_+=n_-\neq 0$. If this is true and $n_+<\infty$, then all the self-adjoint extensions of $T$ are parametrized by $n_+^2$ real parameters. It is therefore very natural to begin this review on self-adjoint extensions of Dirac-Coulomb operator with the computation of the deficency indices. 

\begin{theorem}[\cite{Weid}, Theorem 6.9]
Let $T$ be the Dirac operator with Coulomb potential with $V(x)= \frac{\nu}{|x|}I_4$ defined on $\mathcal{C}^\infty_c$. Then the deficency indices are
\begin{itemize}
\item[i)] $(0,0)$ if $|\nu| \leq \frac{\sqrt{3}}{2}$;
\item[ii)] $(2n(n+1),2n(n+1))$ if $\sqrt{n^2-\frac{1}{4}} < |\nu| \leq \sqrt{(1+n)^2-\frac{1}{4}}$ with $n \in \mathbb{N}$.
\end{itemize}
\end{theorem}

\begin{remark}\label{remark:1}
The deficency indices for the Dirac operator with scalar potential are the same even if we relax the hypothesis of spherical symmetry of the potential. In fact, the statement of the theorem remains unchanged except for the fact that the inequalities become all strict. In order to compute the deficency indices for $\nu=\sqrt{n^2-\frac{1}{4}}$ in the general case of non spherical symmetry one needs additional information on the potential (see \cite{X99} Theorem 4.2).
\end{remark}

\begin{proof}
By passing to polar coordinates and denoting with $\D\Omega$ the surface measure of the  unit sphere we obtain an isomorphism
\begin{equation}
U: L^2(\mathbb{R}^3,\mathbb{C},\D^3x) \to L^2((0,\infty),\mathbb{C},\D r) \otimes L^2(\mathbb{S}^2,\mathbb{C},\D \Omega)
\end{equation}
by setting for each $\Psi \in L^2(\mathbb{R}^3,\mathbb{C},\D^3x)$
\begin{equation}
(U\Psi)(r,\vartheta,\varphi)=r \Psi (\vec x (r,\vartheta,\varphi)).
\end{equation}
The isomorphism can be extended to $\mathcal{L}^2$ component-wise and it will be denoted with the same symbol. 
Under this transformation the free Dirac operator takes the form
\begin{equation}
UT_0U^*= -\I (\boldsymbol\alpha \cdot \hat{\vec{r}})\left(\frac{\partial}{\partial r} - \frac{2}{r} \vec{S}\cdot \vec{L} \right) + \beta.
\end{equation}

Here $\vec L=\vec{x}\times\vec{p}$ denotes the angular momentum operator, $\vec S=-\frac{1}{4} \boldsymbol{\alpha} \times \boldsymbol{\alpha}$ the spin operator and $\hat{\vec{r}}$ is the radial versor.

A direct computation shows that the operator $\vec{S}\cdot \vec{L}$ commutes with the free Dirac operator $UT_0 U^*$. To proceed in the analysis it is convenient to introduce the operator $K=2 \beta(\vec{S}\cdot\vec{L}+1)$ in order to re-write the free Dirac operator as
\begin{equation}
UT_0U^*= -\I (\boldsymbol\alpha \cdot \hat{\vec{r}})\left(\frac{\partial}{\partial r} + \frac{1}{r}-\frac{1}{r}\beta K \right) + \beta.
\end{equation}

Denoting with $\vec J = \vec L + \vec S$ the total angular momentum operator, it is possible to show that the operator $K$ commutes with  $J^2$ and with the third component of the total angular momentum operator $J_3$. Moreover, it is possible to find a common basis of infinitely differentiable orthonormal eigenfunctions on $L^2(\mathbb{S}^2,\mathbb{C}^4,\D\Omega)$ and to prove that all these operators have pure point spectrum (see \cite{Weid} appendix to section 1).

The Hilbert space decomposes into the direct sum of $2$-dimensional spaces
\begin{equation}\label{eq:DecompositionHilbert}
L^2(\mathbb{S}^2,\mathbb{C}^4,\D \Omega) = \bigoplus_{j \in \mathbb{N}+\frac{1}{2}} \bigoplus_{m_j=-j}^j \bigoplus_{\kappa_j=\pm(j+\frac{1}{2})} \mathcal{K}_{m_j,\kappa_j},
\end{equation}
where $\mathcal{K}_{m_j,\kappa_j}=\mathrm{span}\{ \Phi_{m_j,\kappa_j}^+, \Phi_{m_j,\kappa_j}^-\}$ and $\Phi^\pm_{m_j,\kappa_j}$ are smooth common eigenfunctions of $J^2$, $K$, $J_3$ with eigenvalues $j(j+1)$, $\kappa_j$ and $m_j$ respectively. 

Now each vector $\psi \in U^*\mathcal{C}^\infty_c$ can be written as
\begin{equation}
\psi(r,\vartheta,\varphi)= \frac{1}{r}\sum_{j,m_j,\kappa_j} \left(f^+_{m_j,\kappa_j}(r) \Phi^+_{m_j,\kappa_j}(\vartheta, \varphi)+f^-_{m_j,\kappa_j}(r) \Phi^-_{m_j,\kappa_j}(\vartheta,\varphi)\right)
\end{equation}
with coefficient functions $f^\pm_{m_j,\kappa_j}(r) \in C^\infty_c((0,\infty))$. Hence, by putting together all the ingredients, we can compute the action of the radial Dirac-Coulomb operator on each reducing subspace $\mathcal{K}_{m_j,\kappa_j} \otimes C^\infty_c((0,\infty))$ as
\begin{equation}\label{eq:RadialDirac}
t_{m_j,\kappa_j}=\left(\begin{array}{cc}
1+\frac{\nu}{r} & -\frac{d}{dr}+\frac{\kappa_j}{r} \\
\frac{d}{dr}+\frac{\kappa_j}{r} & -1+\frac{\nu}{r} \\
\end{array}\right),
\end{equation}
and one is left with the computation of the deficency indices for this ordinary differential operator. 

Following Weidmann's argument we exploit a limit-point/limit-circle analysis.  The differential operator $t_{m_j,\kappa_j}$ is said to be in the \emph{limit point case} at $0$ (resp. at $\infty$) if for every $\lambda \in \mathbb{C}$ all solutions of $(t_{m_j,\kappa_j}-\lambda)u=0$ are square integrable in $(0,1)$ (resp. in $(1,\infty)$). The operator $t_{m_j,\kappa_j}$ is said to be in the \emph{limit circle case} at $0$ (resp. at $\infty$) if for every $\lambda \in \mathbb{C}$ there is at least one solution of $(t_{m_j,\kappa_j}-\lambda)u=0$ which is not square integrable in $(0,1)$ (resp. in $(1,\infty)$).

Once we know if $t_{m_j,\kappa_j}$ is in the limit circle case or in the limit point case the following general theorem gives us the deficency indices. 

\begin{theorem}[\cite{Weid}, Theorem 5.7] \label{th:Weid57}
The deficency indices of $t_{m_j,\kappa_j}$ are
\begin{itemize}
\item[i)] $(2,2)$ if $t_{m_j,\kappa_j}$ is in limit circle case at both $0$ and $\infty$;
\item[ii)] $(1,1)$ if $t_{m_j,\kappa_j}$ is in limit circle case at one end point and in limit point case at the other;
\item[iii)] $(0,0)$ if $t_{m_j,\kappa_j}$ is in limit point case at both $0$ and $\infty$.
\end{itemize}
\end{theorem}

By Weyl's alternative theorem (see \cite{Weid} Theorem 5.6), either for every $\lambda \in \mathbb{C}$ all solutions of $(t_{m_j,\kappa_j}-\lambda)u=0$ are square integrable in $(0,1)$ (resp. in $(1,\infty)$), or for every $\lambda \in \mathbb{C}\setminus\mathbb{R}$ there exists a unique (up to a multiplicative constant) solution $u$ of $(t_{m_j,\kappa_j}-\lambda)u=0$ which is square integrable in $(0,1)$ (resp. in $(1,\infty)$). Therefore, since no third option is possible, it is sufficient to check whether both the solutions of $t_{m_j,\kappa_j}u=0$ are square integrable in $(0,1)$ and $(1,\infty)$.

To check if $t_{m_j,\kappa_j}$ is in the limit-point or in the limit-circle case we consider the operator $(t_{m_j,\kappa_j}-\beta)$.  The subtraction of a bounded operator does not change the computation of the deficency indices. Choosing $\lambda=0$, the equation to be solved is $(t_{m_j,\kappa_j}-\beta) u=0$. Its solutions are
\begin{equation}
u(r)= \left(\begin{array}{c} u^+(r) \\ u^-(r) \end{array}\right)=\left(\begin{array}{c} \pm \sqrt{\kappa_j^2-\nu^2}-\kappa_j \\
\nu \end{array}\right) r^{\pm \sqrt{\kappa_j^2-\nu^2}}.
\end{equation}

From this explicit expression we see that the solution with positive exponent cannot be square integrable in $(1,\infty)$ and hence independently of the parameters $\kappa_j$ and $\nu$ the operator is always in the limit point case at infinity.

The solution with positive exponent is always square integrable in $(0,1)$ while the one with negative square root is square integrable near zero if and only if
\begin{equation}
-2\sqrt{\kappa_j^2-\nu^2} \leq -1,
\end{equation}
which means
\begin{equation}\label{eq:LPCAtOrigin}
\nu^2 \leq \kappa_j^2-\frac{1}{4}.
\end{equation}
Then if $\nu$ satisfies (\ref{eq:LPCAtOrigin}), the operator is in the limit point case at both endpoints. By Theorem \ref{th:Weid57}, if (\ref{eq:LPCAtOrigin}) holds the deficency indices of the operator are $(0,0)$, otherwise the deficency indices of the operator are $(1,1)$.

To compute the deficency indices of the full operator we have to count how many reduced operators are not essentially self-adjoint. Explicitly,
\begin{equation}
n_{\pm}=\sum_{j \in \mathbb{N}+\frac{1}{2}} \sum_{m_j=-j}^j \sum_{\kappa_j=\pm(j+\frac{1}{2})} \left\{\begin{array}{l}
1 \qquad \mathrm{if} \quad \nu^2 > \kappa_j-\frac{1}{4} \\
0 \qquad \mathrm{else} 
\end{array}\right. .
\end{equation}
Let $n$ be the integer such that $n^2 -\frac{1}{4} < \nu^2 \leq (n+1)^2-\frac{1}{4}$. We obtain
\begin{equation}
n_{\pm}=\sum_{j \in \mathbb{N}+\frac{1}{2}}^{n-\frac{1}{2}} \sum_{m_j=-j}^j 2 = 2n(n+1).
\end{equation}
\end{proof}

\section{Potential with Coulomb-like singularity}\label{sec:History}
In this section we review the historical path that, together with the results in the previous section, led to the present understanding on the existence and uniqueness of self-adjoint extensions of the minimal Dirac-Coulomb operator. 

Conceptually and historically the two main questions addressed so far, and that we are going to analyse are:
\begin{enumerate}
\item Is the operator $T_0+V$ essentially self-adjoint?
\item If it is not, is there a \emph{special} self-adjoint extension which is physically relevant?
\end{enumerate}

The technique employed in answering the first question is essentially a perturbative argument based on the Kato-Rellich theorem and it is addressed in the first subsection.

The second question presents a wider range of answers and many authors provided different meaningful special extensions. Only at a later stage they recognized that, under some hypothesis, they were referring to the same operator. This subject is addressed in the second subsection.

\subsection{Essential self-adjointness via Kato-Rellich theorem}
One of the first proofs of the essential self-adjointness for the Dirac-Coulomb operator is due to Kato in 1951 as a direct application of the Kato-Rellich theorem. Despite the simplicity of the proof, this does not cover the whole range of the parameter $\nu$ on which the Dirac-Coulomb operator is essentially self-adjoint.

Some years later two different approaches based on the same theorem were developed in order to cover the range $[0,\frac{\sqrt 3}{2}]$: the first one, due to Rejt\"o and Gustafsson \cite{R70,GR73} aimed to weaken its hypotheses, the other one due to Schminke \cite{S72} uses the original theorem. Instead of looking to $V$ as a perturbation of $T_0$ he introduced an operator $C$ and considered $T_0+V=(T_0+C)+(V-C)$. To prove the essential self-adjointness of $T_0+V$ he proved separately the essential self-adjointness of $T_0+C$ and looked at $V-C$ as a perturbation satisfying the hypothesis of Kato-Rellich.

Several other works dealt with the same problem, among which we mention \cite{R53,RS62,E70,W71,C77,LR81}. For a self-contained conceptual review we  present in detail only the above-mentioned ones of Rejt\"o-Gustaffson and Schmincke.

Since it will play a central role in this subsection, we recall the classical statement of the Kato-Rellich theorem.

\begin{theorem}[Kato-Rellich]
Suppose that $A$ is an essentially self-adjoint operator, $B$ is a symmetric operator that is $A$-bounded with relative bound $a<1$, namely 
\begin{itemize}
\item[i)] $\mathcal{D}(B) \supset \mathcal{D}(A)$;
\item[ii)] For some $a<1$, $b \in \mathbb{R}$ and for all $\varphi \in \mathcal{D}(A)$,
\begin{equation}
\Vert B\varphi \Vert \leq a \Vert A \varphi \Vert+b\Vert \varphi \Vert.
\end{equation}
\end{itemize}
Then $\overline{A+B}$ is self-adjoint on $\mathcal{D}(\bar{A})$ and essentially self-adjoint on any core of $A$.
\end{theorem}

Let us start with surveying Kato's proof from \cite{K55,K}. The starting point is the well-known Hardy inequality (see \cite{RS} section X.2 p.169)
\begin{equation}
\Vert \vec p u\Vert^2 \geq \frac{1}{4}\Vert r^{-1} u \Vert^2, \qquad \forall u \in C^\infty_c(\mathbb{R}^3).
\end{equation}
By using the properties of the $\boldsymbol\alpha$ matrices we get the identity
\begin{equation}
\Vert T_0 u \Vert^2 = \Vert \vec p u \Vert^2 + \langle (\beta \boldsymbol\alpha \cdot \vec p + \boldsymbol\alpha \cdot \vec p \beta)u,u\rangle +\Vert u \Vert^2 = \Vert \vec p u \Vert^2 + \Vert u \Vert^2.
\end{equation}
Thus, we see that if the potential is $|\phi(\vec x)| \leq \frac{\nu}{|\vec x|}$, we get the following chain of inequalities
\begin{equation}
\Vert \vec p u \Vert^2 \geq \frac{1}{4} \Vert r^{-1} u \Vert^2 \geq \frac{1}{4 \nu^2} \Vert \phi(\vec x) u \Vert^2,
\end{equation}
from which it follows that
\begin{equation}
\Vert \phi(\vec x) u \Vert \leq 4\nu^2 \Vert T_0 u \Vert^2 -4 \nu^2 \Vert u \Vert^2.
\end{equation}
If $\nu < \frac{1}{2}$, the hypotheses of the Kato-Rellich theorem are satisfied and one deduces that $T_0+V$ is essentially self-adjoint and the domain of the unique self adjoint extension is
\begin{equation}
\mathcal{D}(\overline{T_0+V})=\mathcal{H}^1.
\end{equation}

\begin{remark}
By using W\"ust theorem (see \cite{RS} theorem X.14) one can cover the case $\nu=\frac{1}{2}$. However the information on the domain of the self-adjoint extension is lost.
\end{remark}

\begin{remark}
The result is independent of the possible spherical symmetry and of precise matricial form of the potential: the conclusion holds if $\lim_{x \to 0} |x||V_{ij}(x)| < \frac{1}{2}$, where $i,j=1,2$ and $V_{ij}$ are the entries of the matrix $V$.
\end{remark}

\begin{remark}
Arai \cite{A75,A83} showed that by considering more general matrix-valued potentials of the form
\begin{equation}
V(x)=\frac{Z}{r}I_4+\frac{\I}{r} \boldsymbol\alpha \cdot \hat{\vec{r}} \beta b_1 + \frac{\beta}{r} b_2
\end{equation}
the necessary and sufficient condition for the essential self-adjointness is $(\kappa_j+b_1)^2+b_2^2 \geq Z^2 +\frac{1}{4}$ and hence the threshold $\frac{1}{2}$ is optimal, in the sense that if $V$ is in the form above and one of the entry of the matrix satisfies $|x| |V_{ij}|>\frac{1}{2}$ then it is possible to choose $Z,b_1,b_2$ such that the operator is not essentially self-adjoint. 
\end{remark}

In a work from 1970, Rejt\"o \cite{R70} discussed the particular case of spherically symmetric Coulomb-like potentials. By denoting with $B(\mathcal{L}^2)$ the set of bounded operators on $\mathcal{L}^2$, the requirement on $V$ for the operator $T_0+V$ to be essentially self-adjoint on $\mathcal{C}^\infty_c$ boils down to asking that $\exists \mu_{\pm}$ in the upper/lower closed complex half plane such that
\begin{equation}
(1-\bar V (\mu_{\pm}-\bar A_0)^{-1}) \in B(\mathcal{L}^2).
\end{equation}
Proving that the Dirac operator with Coulomb interaction satisfies this hypothesis for $\nu \in [0,\frac34)$, he was able to show that under this condition such an operator is essentially self-adjoint and the domain of its self-adjoint extension is $\mathcal{H}^1$.

In fact \cite{R70} provides some sort of intermediate results that led to the more relevant work \cite{GR73} by Gustaffson and Rejt\"o.  In this relevant continuation they generalized further Kato-Rellich theorem and they were able to achieve the essential self-adjointness for the Dirac operator in the regime $\nu \in [0,\sqrt{3}/2)$. 

Their generalization relies on Fredholm's theory, that we briefly recall here for the self-consistency of the presentation. A densely defined operator $A$ in a Banach space $\mathcal{X}$ is said to be Fredholm if $A$ is closed, $\mathrm{ran} \, A$ is closed, and both $\dim \ker A$ and $\dim \mathcal{X}/\mathrm{ran} \, T$ are finite. The index of a Fredholm operator $A$ is the number $i(T)=\dim \ker A - \dim \mathcal{X}/\mathrm{ran} \, A$.

\begin{theorem}[\cite{GR73}, Theorem 3.1, Generalized Kato-Rellich theorem]
Let $T_0$ be essentially self-adjoint, $V$ symmetric with $\mathcal{D}(V) \supset \mathcal{D}(T_0)$ where $V$ is $T_0$-bounded. For each $\mu$ in the resolvent set of $T_0$ define the operator $A_\mu \in B(\mathcal{L}^2)$ by
\begin{equation}
A_\mu := I-\bar V (\mu - \bar T_0)^{-1}.
\end{equation}
Then the three conditions below
\begin{itemize}
\item[i)] $T_0+V$ is essentially self-adjoint;
\item[ii)] $\overline{T_0+V}=\bar T_0+\bar V$;
\item[iii)] $\mathcal{D}(\overline{T_0+V})=\mathcal{D}(\bar T_0)$;
\end{itemize}
hold if and only if there exists $\mu_+$ in the closed upper half plane and $\mu_-$ in the closed lower half plane such that the operators $A_{\mu_{\pm}}$ are Fredholm of index zero.
\end{theorem}
\begin{proof} (Sketch)
We start from the identity
\begin{equation}
\mu - \bar T_0 - \bar V = [I - \bar V(\mu-\bar T_0)^{-1}](\mu-\bar T_0).
\end{equation}
Since $\mu \in \rho(T_0)$, $\mu-\bar T_0$ is Fredholm of index zero and since the composition of Fredholm operators is Fredholm and the index of the composition is the sum of the indices, by using a standard criterion of essential self-adjointness, we prove the sufficient condition.

The necessity follows using the same index-formula and the fact that if $A_1A_2$ is Fredholm with $A_2$ Fredholm and $A_1$ closed, then $A_1$ is Fredholm and therefore by the above formula $A_{\pm i}$ is Fredholm of index 0.
\end{proof}

\begin{remark}
This theorem includes the classical Kato-Rellich noting that with $\mu_{\pm}=\pm i\frac{a}{b}$ one has $\Vert \bar V(\mu_{\pm}-\bar T_0)^{-1} \Vert < 1$. Hence $A_{\mu_{\pm}}$ are invertible and therefore Fredholm of index zero.
\end{remark}

The proof of the essential self-adjointness of the Dirac operator with Coulomb potential uses the following corollary:
\begin{corollary}\label{cor:GR73}
If there exist $\mu_+$ and $\mu_-$ as in the previous theorem such that $A_{\mu_{\pm}}=B_{\pm}+C_{\pm}$ where $B_{\pm}^{-1} \in B(\mathcal{L}^2)$ and $C_{\pm}$ are compact, then $T_0+V$ is essentially self-adjoint and $\mathcal{D}(\overline{T_0+V})=\mathcal{D}(\bar T_0)$.
\end{corollary}
\begin{proof}
This corollary follows from the fact that an invertible operator is Fredholm of index zero and that this property is stable under compact perturbations.
\end{proof}

By using the spherical symmetry and the decomposition of the Dirac operator Rejt\"o and Gustaffson prove that for $|\nu| \in [0,\frac{\sqrt{3}}{2})$ the hypothesis of Corollary \ref{cor:GR73} are satisfied and hence the spherically symmetric Dirac-Coulomb operator is essentially self-adjoint for that range of parameters. 


In this respect the work of Schmincke \cite{S72} is of interest in that the same conclusion on essential self-adjointness was obtained \emph{independently} of the spherical symmetry of the potential.
\begin{theorem}[\cite{S72}]
Let $\phi \in L^2_{loc}(\mathbb{R}^3\setminus\{0\},\mathbb{R},\D^3 x)$ be a real-valued function that can be expressed as $\phi=\phi_1+\phi_2$ with $\phi_1 \in C^0(\mathbb{R}^3 \setminus\{0\}, \mathbb{R})$ and $\phi_2 \in L^\infty(\mathbb{R}^3 \setminus \{0\},\mathbb{R},\D^3 x)$  with
\begin{equation}
|\phi_1(\vec x)| \leq \frac{\nu}{|\vec x|}
\end{equation}
and $\nu \in [0,\frac{\sqrt{3}}{2})$. Then $T_0+V$ is essentially self-adjoint.
\end{theorem}
The way Schmincke proves its result consists of using the standard Kato-Rellich theorem. He introduces a certain intercalary operator $C$ in order to write $T_0+V=(T_0+C)+(V-C)$ and to regard $V-C$ as a small  perturbation of $T_0+C$. 

More precisely he continues
\begin{equation}
C:=\frac{1}{4}\left(a-\frac{1}{r}\right) \boldsymbol{\alpha} \cdot \hat{ \vec{r}}, \qquad 1 < a < 3
\end{equation}
and $T_0=\boldsymbol\alpha\cdot \vec p+\beta$. He further introduces a bounded operator $S_2$ on which we omit the details. From these definitions it is clear that for $z \in \mathbb{C}$, $0<|z|<1$,
\begin{equation}
T_0+V=(A+\beta+zC)+(V-zC-S_2)+S_2=F+G+S_2.
\end{equation}
With these definitions Schmincke proves that $\Vert Gu\Vert^2 \leq k\Vert Fu \Vert^2$ with $k<1$ and hence $Gu$ is $F$-bounded with a small bound. One can thus apply Kato-Rellich\footnote{Schmincke used a complex version of Kato-Rellich that deals with closed operators instead of self-adjoint ones. This is necessary because, in general, the $z$ appearing in the proof is not real.} to obtain that $T+V+S_2$ is essentially self-adjoint and, since $S_2$ is a bounded operator, this also implies the essentially self-adjointness of $T$.

\subsection{The distinguished self-adjoint extension}\label{subsec:3.2}
As stated in Theorem \ref{th:1}, in the transitory regime there are infinitely many self-adjoint extensions of the minimal Dirac-Coulomb operator. Before considering their classification the main interest throughout the 1970s was the study of a \emph{distinguished} extension characterized by being the most physically meaningful. The first work that introduced this particular self-adjoint extension is due to Schmincke \cite{S72bis} who obtained this extension by means of a multiplicative intercalary operator. This self-adjoint realisation is physically relevant because its domain is contained in the domain of the potential energy form and hence each function on the domain has a finite expectation value of the potential energy operator.

A second and more explicit construction of a distinguished self-adjoint extension of the minimal Dirac-Coulomb operator was found by W\"ust \cite{W75,W77} by means of cut-off potentials. He built a sequence of self-adjoint operators that converges strongly in the operator graph topology to a self-adjoint extension of the minimal Dirac-Coulomb operator. Remarkably that the domain of this self-adjoint extension is also contained in the domain of the potential energy.

At that point it was not clear whether W\"ust's and Schmincke's self-adjoint extensions were the same or not. The first attempt to look for a distinguished self-adjoint extension with a requirement of uniqueness was made by Nenciu \cite{N76} who found that there exists a unique self-adjoint extension of the minimal Dirac operator whose domain is contained in the domain of the kinetic energy form.

In 1979 Klaus and W\"ust \cite{KW79} proved that in the regime $\nu \in (\frac{\sqrt 3}{2},1)$ if the potential $\phi$ is semi-bounded all the above mentioned distinguished self-adjoint extensions coincide.

Let us start with Schmincke's result.
\begin{theorem}
[\cite{S72bis}, Theorems 2 and 3] Let $\phi \in L^2_{loc}(\mathbb{R}^3\setminus\{0\},\mathbb{R},\D^3x)$ be a real-valued function such that $\phi=\phi_1+\phi_2$ with $\phi_1$ and $\phi_2$ both real valued, $\phi_1 \in C^0(\mathbb{R}^3\setminus\{0\})$, and $\phi_2 \in L^\infty(\mathbb{R}^3, \mathbb{R},\D^3x)$. Let $s \in[0,1)$. Suppose there exists $k>1$, $c>1$ and $f \in C^1((0,\infty))$ positive valued and bounded from above by $\frac{1-s}{2c}$ such that
\begin{equation}
\begin{split}
\frac{1}{r^2}\left(f(r)+\frac{s}{2}\right)^2 &\leq k\left(|\phi_1(\vec x)|^2+\frac{1}{r^2}f^2(r)\right) \leq\\
&\leq \frac{1}{r^2}\left(f(r)+\frac{s+1}{2}\right)^2+\frac{1}{r}f'(r) .
\end{split}
\end{equation}
Then there exists a bounded symmetric operator $S$ such that
\begin{equation}
T_G:=\left(\overline{r^{-\frac{s}{2}}}\right)\left(\overline{r^{\frac{s}{2}}(T-S)}\right)+\bar S
\end{equation}
is an essential self-adjoint extension of $T$ and $\forall m \in [\frac{1}{2},1-\frac{s}{2}]$
\begin{equation}
\mathcal{D}(\overline{T}_G)=\mathcal{D}(T^*)\cap \mathcal{D}\left(\overline{r^{-m}}\right).
\end{equation}
\end{theorem}

\begin{remark}
Note that in particular $\mathcal{D}(T_G) \subset \mathcal{D}\left(\overline{r^{-1/2}}\right)$, which physically means that all the functions in the domain of this distinguished self-adjoint extension have a finite expectation of the potential energy.
\end{remark}

Schmincke proved this using a multiplicative intercalary operator. If $T=T_0+V$ with $T_0$ essentially self-adjoint and if there exists a symmetric operator $G$ satisfying suitable properties (see Theorem 1 in \cite{S72bis}), then
\begin{equation}
T_G:=\overline{G^{-1}} \, \overline{GT}
\end{equation}
is an essentially self-adjoint extension of $T$.

Noticeably in the case of Coulomb potential the assumptions of the theorem are satisfied when
\begin{equation}
1-4\nu^2 \leq(1-s^2) \leq 4(1-\nu^2),
\end{equation}
which means $\nu<1$.

W\"ust, instead, showed that given a potential $\phi(\vec x) \in C^0(\mathbb{R}^3 \setminus \{0\})$ such that
\begin{equation}
|\phi(\vec x)| \leq \frac{\nu}{|\vec x|} \qquad |\nu| < 1,
\end{equation}
if one fixes a positive constant $c>0$ and defines
\begin{equation}
V_t(\vec x):=\left\{\begin{array}{l}
V(\vec x) \qquad |\vec x| \geq \frac{c}{t} \\
R(\vec x) \qquad |\vec x|<\frac{c}{t},
\end{array}\right.
\end{equation}
where $R$ is chosen such that the components of $V_t$ are continuous functions. If $V_t(\vec x)$ is definitely monotone, the sequence of operators $T_t=T_0+V_t$ $g$-converges to a self-adjoint operator $T_g$ which is a self-adjoint extension of $T$ with the property that
\begin{equation}
\mathcal{D}(T_g) \subset \mathcal{D}\left(\overline{r^{1/2}}\right).
\end{equation}

In 1976 Nenciu \cite{N76} proposed an alternative distinguished self-adjoint extension $T_N$ by requiring this extension to be the unique with the property that all the functions in its domain have finite kinetic energy, namely
\begin{equation}
\mathcal{D}(T_N) \subset \mathcal{D}(|\bar T_0|^{\frac12}).
\end{equation}
The precise result can be stated as follows.

\begin{theorem}[\cite{N76}, Theorem 5.1]\label{thm:Nenciu}
Let $w(t)$ be a decreasing function on $[0,\infty)$ such that $0 \leq w(t) \leq 1$, $\lim_{t \to \infty} w(t)=0$, $T_0$ and $V$ be a matrix-valued potential.

If
\begin{itemize}
\item[i)] $V(\vec x)=w(|\vec x|)W(\vec x)$ where $W$ is a small perturbation of $T_0$, or
\item[ii)] $V(\vec x)=V_1(\vec x)+V_2(\vec x)$ where $V_1$ is dominated by the Coulomb potential with coupling constant $\nu < 1$ and $V_2=w(|\vec x|)W_2(\vec x)$ where $W_2$ is non-singular,
\end{itemize}
then
\begin{itemize}
\item[i)] there exists a unique operator $T_N$ such that 
\begin{equation}
\mathcal{D}(T) \subset \mathcal{D}(|\overline{T}_0|^{\frac12});
\end{equation}
\item[ii)] $\sigma_{ess}(T) \subset \sigma_{ess}(T_0)$.
\end{itemize}
\end{theorem}

The proof relies on a variant of Lax-Milgram lemma and has the inconvenience not to be constructive.

In 1979 Klaus and W\"ust \cite{KW79} showed that in the case of semi-bounded potential W\"ust's and Nenciu's distinguished extensions actually coincide. This is an interesting fact both from a physical and from a mathematical point of view. Physically this coincidence means that the distinguished self-adjoint extension has the property of being the only one whose functions in the domain have finite potential \emph{and} kinetic energy. From a mathematical point of view this overcomes the fact that Nenciu's method was not constructive and provides instead an explicit expression for its self-adjoint extension in terms of $g$-limit of $T_t$.

The identification of a certain \emph{distinguished extension} was pushed further by Esteban and Loss \cite{EL08} up to the value $\nu=1$. In that paper they proposed to define the distinguished self-adjoint extension via Hardy-Dirac inequalities. By a limit argument this procedure can define a sort of distinguished self-adjoint extension also when $\nu=1$ but, in that case, the domain of this self-adjoint extension will be neither contained in the domain of the kinetic energy form nor in the domain of the potential energy form. In a subsequent work, Arrizabalaga \cite{A11} weakened further  the hypothesis on the construction of the self-adjoint extension of Esteban and Loss.

\section{Classification of the self-adjoint extensions}\label{sec:Hogeve}
In the previous section we discussed the distinguished extension of the minimal Dirac-Coulomb operator with respect to an infinite multiplicity of others. We come now to the major problem of classifying the one-parameter family of self-adjoint extensions of such an operator.

There is essentially one work, by Hogreve \cite{H13}, that deals systematically with this problem. There the classification is made by means of von Neumann's extension theory. We start by recalling von Neumann's theorem on general parametrization of symmetric extension.

\begin{theorem}[von Neumann] 
Let $T$ be a densely defined, closed and symmetric operator. The closed symmetric extensions of $T$ are in one to one correspondence with the set of partial isometries (in the usual inner product) of $\ker(T^*+i)$ into $\ker(T^*-i)$. If $U$ is such an isometry with initial space $I(U) \subseteq \ker(T^*+i)$, then the corresponding closed symmetric extension $T_U$ has domain
\begin{equation}
 D(T_U) =\{\varphi + \varphi^{(i)}+U\varphi^{(i)} \, | \, \varphi \in D(T), \, \varphi^{(i)} \in I(U) \}
\end{equation}
and
\begin{equation}
 T_U(\varphi+\varphi^{(i)}+U\varphi^{(i)})=T\varphi+i\varphi^{(i)}-i U \varphi^{(i)}.
\end{equation}
If $\dim I(U) < \infty$, the deficency indices of $T_U$ are
\begin{equation}
 n_{\pm}(T_U)=n_{\pm}(T)-\dim[I(U)].
\end{equation}
\end{theorem}

Recalling that if $\nu^2 > \kappa_j^2 -\frac{1}{4}$ the deficency indices of $t_{m_j,\kappa_j}$ are $(1,1)$, the isometries are just phases: given $\theta \in [0, 2\pi)$ we have
\begin{equation}
\begin{array}{ccccc}
U_\theta & : &\ker(T^*+i) &\to& \ker(T^*-i)\\
		 &   &  \psi^{(i)} & \mapsto & e^{i\theta} \psi^{(-i)}, \\
\end{array}
\end{equation}
and hence at every value of $\theta$ there corresponds one self-adjoint extension of the minimal operator $t_{m_j,\kappa_j}$.  

Let $\psi(r)=(\psi_1(r), \psi_2(r)) \in AC((0,\infty),\mathbb{C}^2)$. We define
\begin{equation}
(\Theta_\theta \psi)(r)=(\psi_2^{(-i)}(r)+e^{i\theta} \psi_2^{(i)}(r), -\psi_1^{(-i)}(r)-e^{i\theta}\psi_1^{(i)}(r)) \cdot \left(\begin{array}{c}
\psi_1(r) \\ \psi_2(r) \end{array}\right).
\end{equation}

\begin{theorem}[\cite{H13}, Theorem 7.1]\label{thm:Hogeve}
Let $|\nu| > \sqrt{\kappa_j^2 -\frac{1}{4}}$, the self-adjoint extensions $t_{m_j,\kappa_j}^\theta$ of the minimal operator $t_{m_j,\kappa_j}$ of (\ref{eq:RadialDirac}) are uniquely determined by $\theta \in [0,2\pi)$ via the formulas
\begin{multline}\label{eq:DomainSA}
\mathcal{D}(t_{m_j,\kappa_j}^\theta)=\{ \psi \in L^2((0,\infty),\mathbb{C}^2) \cap AC((0,\infty),\mathbb{C}^2) \, | \, \lim_{r \to 0} (\Theta_\theta \psi)(r) , \\
 t_{m_j,\kappa_j} \psi \in L^2((0,\infty)) \}
\end{multline}
\begin{equation}
t_{m_j,\kappa_j}^\theta \psi = t_{m_j,\kappa_j}^* \psi.
\end{equation}
\end{theorem}

\begin{proof}We prove only one direction of the theorem.

By  von Neumann's theorem above and the explicit formula for the unitary transformation we have
\begin{equation}
\mathcal{D}(t_{m_j,\kappa_j}^\theta)=\mathcal{D}\left(\overline{t_{m_j,\kappa_j}^\theta}\right)+\{c(\psi^{(i)}+e^{i\theta}\psi^{(-i)}) \, | \, c \in \mathbb{C} \}
\end{equation}
and hence by taking the limit $r \to 0$ one gets
\begin{equation}
c=\lim_{r \to 0} \frac{\psi_n(r)-\phi_n(r)}{\psi_n^{(-i)}+e^{i \theta} \psi_n^{(i)}(r)} 
\end{equation}
with $n=1,2$. This implies that taking into account that $\psi_n\to 0$, for $r \to 0$ the quantity with $n=1$ equals the one with $n=2$ and this is precisely the condition $\lim_{r \to 0}(\Theta_\theta \psi)=0$.
\end{proof}

\section{Future perspectives: a selection of main open problems}\label{sec:Future}
In the final part of these notes we survey a few topical questions concerning the multiplicity of self-adjoint realisations of the model.
\begin{itemize}
\item[i)] \emph{Characterisation of $\mathcal{D}(t^\theta_{m_j,\kappa_j})$}. The sole characterisation of the domains of the self-adjoint extensions present in the literature, namely (\ref{eq:DomainSA}) above, does not give any explicit detail on the behaviour of the functions near the origin. More refined information on this short scale behaviour is expected to be achievable by means of the self-adjoint extension theory of Kre\u{\i}n-Vi\v{s}ik-Birman (KVB) (see for example \cite{M17}).
\item[ii)] \emph{Adaptation of the original extension formulas} for initial operators that are not semi-bounded (as is the case for $t_{m_j, \kappa_j}^\theta$). The original KVB theory is developed in order to classify the self-adjoint extensions of a semi-bounded operator. An operator version of the extension formula for non semi-bounded operators with a spectral gap can be found in \cite{Grubb1,Grubb2} but, to our knowledge, a similar theorem for the corresponding quadratic forms is not available in the literature.
\item[iii)] \emph{Qualification of further features of the domain of the distinguished extension}. Beside the huge amount of studies concerning the domain of the distinguished extension (see for example \cite{A11,ADV13}), the available knowledge on such operator remains somewhat implicit. Among the other informations, one would like to qualify the most singular behaviour at zero of the generic element of the domain, and how this behaviour may depend on the magnitude of the coupling constant $\nu$.
\item[iv)] \emph{General classification of the extensions both in the operator sense and in the quadratic form sense}, where the effectiveness of the classification relies in the possibility of qualifying special subclasses of interest (e.g. invertible ones). In particular, it would be of relevance to reproduce, in analogy to what happens for semi-bounded operators, the natural ordering of the quadratic forms.
\item[v)] \emph{Study of the spectral properties of the generic extension}, with particular focus on the discrete spectrum lying in $[-1,1]$. For example, one would like to identify the self-adjoint realisation with the highest number of eigenvalues or the one with the lowest eigenvalue or one could even try to identify the lowest possible (absolute value of the) eigenvalue among the extensions.
\item[vi)] \emph{Identification of an analogous notion of distinguished extension in the regime $\nu > 1$}. The reason for which if $\nu > 1$ there is no self-adjoint realisation with the property that its domain is contained in the domain of the potential energy form can be seen with the decomposition (\ref{eq:DecompositionHilbert}). If $\nu>1$ the functions in the domain of the reduced operator in the sector with $j=\frac{1}{2}$ do not vanish at $r=0$. In all the other sectors, however, this is not the case. It is thus possible to prove the existence of a special self-adjoint realisation of the reduced Dirac operator in the sectors with $j\geq \frac{3}{2}$ that retains most of the properties of the distinguished extension in the regime $\nu \in [\frac{\sqrt3}{2}, 1)$.
\end{itemize}


\end{document}